\documentclass[12pt]{iopart}

\usepackage[british]{babel}      

\usepackage{amsfonts} 
\usepackage{amssymb}  
\usepackage{amsthm}   

\usepackage{comment}

\usepackage{ae}      

\usepackage{titling}

\newcommand{\R}{\mathbb{R}}
\newcommand{\C}{\mathbb{C}}

\renewcommand{\le}{\leqslant}
\renewcommand{\ge}{\geqslant}

\newcommand{\ra}{\rightarrow}

\usepackage{enumerate}    

\newtheoremstyle{example}{\topsep}{\topsep}%
     {\itshape}
     {}
     {\bfseries}
     {.}
     {\newline}
     {\thmname{#1}\thmnumber{ #2}\thmnote{ #3}}

\theoremstyle{example}
\newtheorem{theorem}{Theorem}[section]

\newtheorem{lemma}[theorem]{Lemma}
\newtheorem{corollary}[theorem]{Corollary}
\newtheorem{definition}[theorem]{Definition}

\newtheoremstyle{remark}{\topsep}{\topsep}%
     {}
     {}
     {\bfseries}
     {.}
     {\newline}
     {\thmname{#1}\thmnumber{ #2}\thmnote{ #3}}

\theoremstyle{remark}
\newtheorem{remark}[theorem]{Remark}

\newtheoremstyle{theorem*}{\topsep}{\topsep}%
      {\itshape}
      {}
      {\bfseries}
      {.}
      {\newline}
      {\thmname{#1}\thmnote{ #3}}

\theoremstyle{theorem*}
\newtheorem{theorem*}{Theorem}

\begin{document}

\title{Trace Formulae for quantum graphs with edge potentials}
\author{Ralf Rueckriemen$^1$  and Uzy Smilansky$^{1,2}$ }
\address{$^{1}$School of Mathematics, Cardiff University, Cardiff CF24 4AG, Wales, UK}
\address{$^{2}$Department of Physics of Complex Systems,Weizmann Institute of Science, Rehovot 76100, Israel.}


\begin{abstract}
This work  explores the spectra of quantum graphs where the Schr\"odinger operator on the edges is equipped with a potential. The scattering 
approach, which was originally introduced for the potential free case, is extended to this case and  used to derive a secular function whose zeros coincide with the eigenvalue spectrum. Exact trace formulas for both smooth  and  $\delta$-potentials are derived, and an asymptotic semiclassical trace formula (for smooth potentials) is presented and discussed.
\end{abstract}

\section{Preliminaries and a summary of the main results}
A metric graph is a finite combinatorial graph where each edge is endowed with the natural metric and with a positive real length. A quantum graph is a metric graph together with a Schr\"odinger operator acting on it. The operator acts as second derivative plus lower order parts on the individual edges. Boundary conditions are imposed at the vertices to ensure that the operator is self adjoint. Quantum graphs where first introduced in the 1930s by Pauling to study the movement of electrons in molecules. They are a popular model for various physical processes involving wave propagation. The survey articles \cite{GnutzmannSmilansky06}, \cite{Kuchment08} and \cite{Post09} provide an excellent overview.

One of the key properties that make quantum graphs such a useful model is the existence of an exact trace formula.
The first trace formula for quantum graphs was proved in \cite{Roth83}. It is valid for the standard Laplacian with Kirchhoff-Neumann boundary conditions at the vertices. Various generalizations have been shown, valid for Schr\"odinger operators with a magnetic field and much more general boundary conditions, (see \cite{KottosSmilansky99, KPS07, BolteEndres09} or \cite{BolteEndres08} for a survey). However, all of these rely on the fact, that the eigenvalue equation can be solved explicitly on each of the individual edges, the eigenfunctions are linear combinations of counter-propagating free waves.

Here, we consider the case of a more general Schr\"odinger operator $\Delta+w$ where $w$ stands for the set of edge potentials, as will be detailed below.  In this case the solutions on the edges are no longer known explicitly. However, Sturm-Liouville
theory still guarantees the existence of two independent solutions on each edge. In the high energy limit these solutions converge to the freely propagating  waves of the no potential setting (see for example \cite{PoeschelTrubowitz87}). Quantum graphs with a potential are also considered in \cite{GnutzmannSmilansky06}, \cite{ChernyshevShafarevich07} and \cite{HKT12}. Their methods are complementary to ours and the focus is 
on different aspects of the problem.

Before stating our main results we provide some necessary definitions and describe the model.

\subsection{The model}
\label{section:model}

Let $\mathcal{G}$ be a graph defined by a vertex set $\mathcal{V}$ and an edge set $\mathcal{E}$. Their respective cardinalities
will be denoted by $V=|\mathcal{V}|$ and $E=|\mathcal{E}|$. The natural
Euclidean measure is attributed to the edges, and the length of an edge is denoted by $L_{e}$, for $e \in \mathcal{E}$. The
number of edges and vertices as well as the length of each edge are assumed to be finite.  Distance
along an edge can be measured from either of the connected vertices, and it is convenient to use a directed edge notation to specify
the direction of increasing coordinate. Thus, given an edge $e\in \mathcal{E}$, one defines the two corresponding directed edges
denoted by $d$ and $\hat{d}$. The set of directed edges will be denoted by $ \mathcal{D} $. The initial and terminal
vertices of a directed edge $d\in  \mathcal{D} $ will be denoted by $ \iota(d)$ and $\tau(d)$, respectively. Clearly
$\iota (d)=\tau(\hat{d})$. The coordinate on the directed edge will be denoted by
$x_{d}$ and it assumes the value $0$ at $\iota(d)$ and $L_{d}$ at $\tau(d)$. Thus,
$x_{d} = L_{d}-x_{\hat{d}}$. It is  convenient to associate the length to the directed bonds, in which
case $L_{d}=L_{\hat{d}}=L_{e}$.

The Schr\"odinger operator on each directed edge is defined as the sum of the (magnetic) Laplacian and a real continuous potential $w_{d}(x)$:
\begin{equation}
\Delta_{w,A} := -\left( \frac{\partial}{\partial x_{d}} - i A_{d}\right )^2 + w_{d}(x), \ \ \ {\rm with}
\ \ d\in  \mathcal{D} \ .
\end{equation}
Here $A_{d}$ are the position independent ``magnetic potentials" associated to each edge with
$A_{d}=-A_{\hat{d}}$. For simplicity of exposition, it will be assumed that $A_d=0$ on all edges, the generalization is straight forward.
Denote the operator $\Delta_{w,0}=\Delta_w$. Clearly,  $w_{d}(x)=w_{\hat{d}}(L_e-x)$.  The edge
operator acts on the space of twice differentiable complex valued functions in the interval $[0,L_e]$.
It is natural to define the edge potentials to be continuous on all the vertices:
$w_{{d}}(0)=w_{{d'}}(0)$ for all $d,d': \iota(d)=\iota(d')$.

The graph Schr\"odinger operator acts on $2E$ dimensional vectors $\left (\psi_d(x)\right )_{d\in \mathcal{D}}$. To get a consistent
description, they are invariant under a switch of the preferred direction on the edges,  $\psi_d(x)=\psi_{\hat{d}}(L_d-x)$. The bond wave
functions are continuous at the vertices and are doubly differentiable  in $(0,L_d)$. The Schr\"odinger operator is self adjoint if
appropriate boundary conditions are chosen at the vertices. For the sake of clarity of notation, only the
Kirchhoff-Neumann boundary conditions will be used:
\begin{eqnarray}
\sum_{d\in \mathcal{D} :\ \iota(d)=v} \left. \frac{\partial \psi_d }{\partial x_d}\right|_{x_d=0} = 0 \qquad\qquad & \forall \ v\in \mathcal{V}\ .
\end{eqnarray}
The formalism can be easily extended to include other types of boundary conditions as in \cite{KostrykinSchrader99}.

It is instructive to recall the main steps taken to
derive the trace formula and point out the differences and similarities between the no potential case as done in
\cite{KottosSmilansky97} and \cite{KottosSmilansky99} and the potential case here.

On every directed edge $d$ the eigenvalue equation
\begin{eqnarray}
 -\psi_{d}''(x)+w_{d}(x)\psi_{d}(x)=k^2 \psi_{d}(x)
\end{eqnarray}
has two independent solutions, denoted by $\psi_{d}^+$ and $\psi_{d}^-$.
Normalize them so that
\begin{eqnarray}
 \psi_d^{\pm}(0)=1 &&\qquad\qquad\qquad (\psi_d^{\pm})'(0)=\mp ik
\end{eqnarray}
This corresponds to the two solutions $e^{\mp ikx}$ in the no potential setting, where the sign convention is made so that $\psi_d^+$ is the
incoming wave and $\psi_d^-$ is the outgoing wave.
The general solution is then given by
\begin{eqnarray}
\label{amplitudes}
 \psi_{d}(x)=a_{d}\psi_{d}^+(x) + b_{d}\psi_{d}^-(x)
\end{eqnarray}
Finding an eigenfunction on the entire graph is now equivalent to finding a list of coefficients $a_{d}$ and $b_d$.
These coefficients have to satisfy two types of conditions:

First, they have to be invariant under a switch of orientation on each edge. The general solution on the edge $e$
can be rewritten in terms of $\psi_{\hat{d}}^+$ and $\psi_{\hat{d}}^-$, the two solutions on the same edge with reversed orientation.
\begin{eqnarray}
 \psi_{d}(L_e-x)=\psi_{\hat{d}}(\hat{x})=a_{\hat{d}}\psi_{\hat{d}}^+(\hat{x}) + b_{\hat{d}}\psi_{\hat{d}}^-(\hat{x})
\end{eqnarray}
Here $\hat{x}$ traverses the edge in the opposite direction.

The resulting linear relation is captured in the edge transition matrix $ t^{(e)} (k;w_e)$ defined by:
\begin{eqnarray}
 \left(\begin{array}{c}a_{d}\\a_{\hat{d}} \end{array}\right)
&=   t^{(e)} (k;w_e) \left(\begin{array}{c}b_{\hat{d}}\\b_{d} \end{array}\right)
\end{eqnarray}

The matrix $t^{(e)} (k;w_e)$ is unitary, it depends on the potential on the edge and on $k$. It will be derived explicitly in section
\ref{section:single_edge}. If there is no potential the edge transition matrix describes  the  phase shift acquired by free propagation
along the edge, and $t^{(e)}(k;0)$ is a diagonal matrix with equal diagonal entries $e^{ikL_e}$.
If there is a potential, there will be both transmission and reflection on the edge, and the edge transition matrix is not necessarily diagonal.

Second, the eigenfunctions have to satisfy the boundary conditions at all the vertices, this step is equivalent to the no potential case and
is captured in the vertex scattering matrices $\sigma^{v}$.
\begin{eqnarray}
 b_{d}=\sum_{d': \tau(d')=v}\sigma^v_{dd'}a_{d'}
\end{eqnarray}
where $\iota(d)=v$.
In physical terms, $\sigma^{v}$ describes the scattering of waves at a vertex $v$ as dictated by the vertex boundary conditions.
For the Kirchhoff-Neumann boundary conditions it reads,
\begin{eqnarray}
\sigma^{v}_{dd'}=-\delta_{{\hat d}, d' }+ \frac{2}{\deg(v)},\ \ \ {\rm for}  \ \ \ \tau(d') = \iota (d) = v \
\end{eqnarray}
where $\deg(v)$ is the degree (valency) of the vertex $v$ and $\delta$ is the Kronecker-$\delta$.  Note that the scattering from an incoming edge to all outgoing edges is the same except for the scattering back to the same edge.

The vertex scattering matrices are the building blocks for construction of the $2E\times 2E$ matrix
 \begin{equation}
 \label{def Sigma}
\Sigma_{dd'}=\delta_{\tau(d'),\iota(d)}\sigma^{v}_{dd'}\ \ {\rm where} \ \ v=\iota(d)\ .
\end{equation}
 In a similar way, the graph transition matrix $T(k;w)$ is constructed from the edge transition matrices
\begin{equation}
 \label{def T}
 T_{dd'}(k;w) = (\delta_{{\hat d}, d'}+\delta_{d,d'}) t^{(e)}_{dd'}(k;w_e),\ \  {\rm where}\ \ e=(\iota(d),\tau(d))\ .
\end{equation}

\subsection{The main results}
\label{subsection:main}

Putting all the above conditions and definitions together, gives rise to the {\it secular function} $\zeta_w : \mathbb{C}\rightarrow \mathbb{C} $.
The first main result of this work states its explicit form:

\begin {theorem}
\label{theorem-secular}
Let $\mathcal{G}$ be a metric graph with the Schr\"odinger operator $\Delta_w$. Then there exists a secular function of the form
\begin{equation}
\label{secular}
\fl \qquad \zeta_w (k) = (\det S_w(k))^{-\frac{1}{2}}\det (Id_{2E}-S_w(k)),\qquad {\rm where}\qquad S_w(k)=\Sigma T(k;w)\ ,
\end{equation}
whose zeros correspond to the eigenvalues of the operator including multiplicities (with the possible exception of the eigenvalue zero).
\end {theorem}
Note that the secular function is of the same form as in the no potential case, see for example \cite{KottosSmilansky99}. The vertex scattering matrix $\Sigma$ is identical. The key difference is in the transition matrix $T(k;w)$. If
there is no potential this is a diagonal matrix which takes care of the  phase shift accumulated when the wave propagates from one end of an edge to the other. In the present setting it is not necessarily diagonal,
and it describes the transmission and  the reflection induced by the edge potential.

The Cauchy argument principle can be applied to count the zeros of the secular function (\ref{secular}) and to derive a trace formula. However, before this is done, the concept of periodic orbits on the graph should be re-examined.
 
In general, a periodic orbit is a closed oriented walk on the graph. It is completely determined by the topology (connectivity) of the graph. The periodic orbits on a potential free quantum graph coincide with the topological periodic orbits. However, in the presence of an edge potential, a classical particle or a wave is not only  transmitted through the edge but can also be reflected back. Thus, a periodic orbit can include back-scattering from the edge potential - a dynamically induced variant which is not accounted for by the topology.  In order to systematically include the possibility of potential induced reflections, the following scheme is proposed. It makes use of an auxiliary graph denoted by $\mathcal{G}^*$.  It is created from the original graph
$\mathcal{G}$ by inserting an auxiliary vertex (of degree 2) on each edge. The resulting graph  is bipartite. Now, the reflection  and transmission induced by the $2\times2$ transition matrices $ t^{(e)} (k;w_e)$ can be formally considered as the vertex scattering matrix of the corresponding auxiliary vertex. The set of periodic orbits in the trace formula is then the
topological set of periodic orbits on $\mathcal{G}^*$, which will be denoted by $PO^*$. This is a new feature due to the edge potentials that does not occur in the free setting. 

The trace formula is a formal equality between distributions.  To give it a meaning as a functional equality, a suitable space of test functions is required. Throughout, test functions are denoted by $\varphi$, it is assumed they are analytic on the real line, admit a holomorphic extension to a neighbourhood of the real line and are rapidly decreasing.

The second main result of the present work is given in the following theorem.
\begin{theorem}
\label{theorem-trace}
Let $\mathcal{G}$ be a metric graph with Schr\"odinger operator $\Delta_w$. Denote its positive eigenvalues
by $\{k_n^2\}_n$. Then there exists $K \ge 0$ such that
\begin{eqnarray}
\label{exact_trace}
\fl \qquad \sum_{k_n>K }^{\infty} \varphi(k_n)
=& \frac{1}{2\pi }\int_{K}^{\infty}\varphi(k)\frac{\partial \Theta_w(k)}{\partial k}dk
 +\frac{1}{\pi }Im\int_{K}^{\infty}\varphi(k)\sum_{p \in PO^*} \mathcal{A}_p(k)dk
\end{eqnarray}
for all test functions $\varphi$. The phase of  $\det T(k;w)$ is denoted by
$\Theta_w(k)$. The $\mathcal{A}_p(k)$ are complex amplitudes associated to each periodic orbit. They are build from elements of the matrix
$S_w(k)$ along the periodic orbit, the explicit formula is written down in (\ref{A-def}).
\end{theorem}

The threshold $K$ depends on the edge potentials and it is a necessary ingredient in the theory. Its origin will be explained     both in the formal proof of the theorem and in the example where the trace formula for $\delta$-potentials is derived.  The trace formula (\ref {exact_trace}) is exact and the coefficients $\Theta_w(k)$ and $\mathcal{A}_p(k)$ depend  implicitly on the potential. For high energies
the eigenfunctions converge to the free solutions. Thus to leading order in $k$ the trace formula above converges to the no potential case. Indeed, in the potential free case, $K=0$, $\Theta_0(k)=2k\mathcal{L}$, where $\mathcal{L}=\sum_{e \in \mathcal{E}}L_e$ and the leading order part of $\mathcal{A}_p(k)$
depends only on the vertex scattering matrix $\Sigma$ and the edge lengths. Using the WKB-approximation it is also possible to compute the first order correction term. The resulting semi-classical trace formula will be discussed in the last section of the present work.

The rest of the paper is arranged as follows.  In section \ref{section:single_edge}
the transition matrix for a single edge $t (k;w)$ will be explicitly
derived, and its relevant properties will be discussed. Using this information, the proof of the two main theorems
will follow in section
\ref {section:trace}. The construction will be illustrated in section \ref{section:illustrate} by considering a
simple system with a $\delta$-potential on the edges. Finally, the asymptotic semi-classical trace formula
will be derived in section \ref{section:high_energy}.

\section {The transition matrix for a single edge}
\label{section:single_edge}

Consider a single edge $e$ of the quantum graph. As explained in section \ref{section:model} the general solution of the eigenvalue equation
can be expressed in two different ways corresponding to the different orientations of the edge.
\begin{eqnarray}
\label{two_way_expression}
 a_{d}\psi_{d}^+(x) + b_{d}\psi_{d}^-(x)=a_{\hat{d}}\psi_{\hat{d}}^+(L-x) + b_{\hat{d}}\psi_{\hat{d}}^-(L-x)
\end{eqnarray}

Recall the definition of the edge transition matrix.

\begin{definition}
The edge transition matrix $t^{(e)} (k;w_e)$ is defined by the equation
\begin{eqnarray}
 \left(\begin{array}{c}a_{d}\\a_{\hat{d}} \end{array}\right)
&=   t^{(e)} (k;w_e) \left(\begin{array}{c}b_{\hat{d}}\\b_{d} \end{array}\right)
\end{eqnarray}
 \end{definition}

Using the symbol $'$ to denote  differentiation with respect to $x$, equation (\ref{two_way_expression}) implies
\begin{eqnarray}
 \psi_{d}(0)=\psi_{\hat{d}}(L) &&\qquad\qquad\qquad \psi_{d}'(0)=-\psi_{\hat{d}}'(L) \ .
\end{eqnarray}

A direct computation then yields an explicit expression for the transition matrix.

\begin{lemma}
\label{T-matrix}
The transition matrix is given by
\begin{eqnarray}
\fl t^{(e)} (k;w_e)=&  \frac{1}{\psi_{d}^+(0)(\psi_{\hat{d}}^+)'(L) + (\psi_{d}^+)'(0)\psi_{\hat{d}}^+(L)}  \\
& \cdot \left(\begin{array}{cc} -W(\psi_{\hat{d}}^+, \psi_{\hat{d}}^-) & -\psi_{d}^-(0)(\psi_{\hat{d}}^+)'(L) - (\psi_{d}^-)'(0)\psi_{\hat{d}}^+(L)\\
 - \psi_{d}^+(0)(\psi_{\hat{d}}^-)'(L)-(\psi_{d}^+)'(0)\psi_{\hat{d}}^-(L)    &  -W(\psi_{d}^+, \psi_{d}^-) \end{array}\right) \nonumber
\end{eqnarray}
where $W(\psi_{\hat{d}}^+, \psi_{\hat{d}}^-)=\psi_{\hat{d}}^+ (\psi_{\hat{d}}^-)'-(\psi_{\hat{d}}^+)'\psi_{\hat{d}}^-$ is the Wronskian.
\end{lemma}

In order to derive the trace formula the matrix $t^{(e)} (k;w_e)$ needs to be unitary on the real line. This is achieved through
the choice of a suitable normalization of $\psi_d^+$ and $\psi_d^-$. The matrix $t^{(e)} (k;w_e)$ is unitary on the real line whenever
\begin{eqnarray}
 \psi_d^+(0)= \overline{\psi_d^-(0)}
\end{eqnarray}
which is satisfied with the normalization $\psi_d^{\pm}(0)=1$ and $(\psi_d^{\pm})'(0)=\mp ik$.
Note that the potential is real, so the system has two real solutions, hence $\psi_d^+(x)= \overline{\psi_d^-(x)}$ for all points
on the interval. It also implies that $W(\psi_{d}^+, \psi_{d}^-) \in i\R$ and $|\psi_d^+|=|\psi_d^-|$.

The eigenvalues of $t^{(e)} (k;w_e)$ are
\begin{eqnarray}
 \mu_{1,2}&=-\frac{W(\psi_{\hat{d}}^+, \psi_{\hat{d}}^-)\pm |\psi_{d}^-(0)(\psi_{\hat{d}}^+)'(L) + (\psi_{d}^-)'(0)\psi_{\hat{d}}^+(L)|}{\psi_{d}^+(0)(\psi_{\hat{d}}^+)'(L) + (\psi_{d}^+)'(0)\psi_{\hat{d}}^+(L)}
\end{eqnarray}
For the derivation of the trace formula it is necessary to have control of the behaviour of  $t^{(e)} (k;w_e)$ in the vicinity of the real $k$ axis. This is done in using the following lemma. 
\begin{lemma}
\label{subunitary}
 For any given potential $w_e \in L^2([0,L])$ there exist a $K$ such that the matrix $t^{(e)} (k+i\varepsilon;w_e)$ is sub-unitary for
$k > K$ and $\varepsilon >0$.

For high energies the matrix $t^{(e)} (k;w_e)$ converges to the no potential case.
\begin{eqnarray}
 t^{(e)} (k;w_e) = \left(\begin{array}{cc} e^{ikL} & 0 \\ 0 & e^{ikL} \end{array}\right) + O(k^{-1})
\end{eqnarray}

\end{lemma}
\begin{proof}
Take $k$  real. From standard properties of Sturm-Liouville problems, \cite{PoeschelTrubowitz87}, the two solutions $\psi^{\pm}$ satisfy
\begin{eqnarray}
 \left| \psi^{\pm}(x) - e^{\mp ikx}\right| &< \frac{2}{k}e^{||w||_{L^2}\sqrt{L}} \nonumber \\
 \left| (\psi^{\pm})'(x) \mp ike^{\mp ikx}\right| &< 2||w||_{L^2}e^{||w||_{L^2}\sqrt{L}}
\end{eqnarray}
and
\begin{eqnarray}
 \frac{\partial \psi^{\pm}}{\partial k}(x)
=& -i\int_0^x \psi^{\pm}(\tilde{x})\left(\psi^+(x)\psi^-(\tilde{x})-\psi^+(\tilde{x})\psi^-(x)\right)d\tilde{x} \nonumber \\
& \mp \frac{1}{2k}\left(\psi^+(x)-\psi^-(x)\right)\nonumber \\
 \frac{\partial (\psi^{\pm})'}{\partial k}(x)
=& -i\int_0^x \psi^{\pm}(\tilde{x})\left((\psi^+)'(x)\psi^-(\tilde{x})-\psi^+(\tilde{x})(\psi^-)'(x)\right)d\tilde{x} \nonumber \\
 &\mp \frac{1}{2k}\left((\psi^+)'(x)-(\psi^-)'(x)\right)
\end{eqnarray}
 Define an error function $E(x,k)$ via
\begin{eqnarray}
\label{error_function}
 \psi^{+}(x)= e^{- ikx}+ E(x,k) &&\qquad\qquad\qquad
\psi^{-}(x)= e^{ikx}+ \overline{E(x,k)}
\end{eqnarray}
then the bounds above imply
\begin{eqnarray}
 \left| E(x,k) \right| < \frac{C}{k} &\qquad\qquad
\left| E'(x,k) \right| < C \nonumber \\
\left| \frac{\partial}{\partial k}E(x,k) \right| < \frac{C}{k} &\qquad\qquad
\left| \frac{\partial}{\partial k}E'(x,k) \right| < C
\end{eqnarray}
for some constant $C$ depending on the potential but independent of $k$ and $x \in [0,L]$.
Plugging (\ref{error_function}) in the transition matrix yields
\begin{eqnarray}
\fl \qquad t^{(e)} (k;w_e)=&  \frac{1}{e^{-ikL} - \frac{1}{2ik}\left( E'(L,k)-ikE(L,k)\right)} \nonumber  \\
& \cdot \left(\begin{array}{cc} 1 & \frac{1}{2ik}\left(E'(L,k)+ikE(L,k)\right)\\
 \frac{1}{2ik}\left(\overline{E'(L,k)} - ik\overline{E(L,k)}\right)    &  1 \end{array}\right)
\end{eqnarray}
This implies
\begin{eqnarray}
 t^{(e)} (k;w_e) = \left(\begin{array}{cc} e^{ikL} & 0 \\ 0 & e^{ikL} \end{array}\right) + O(k^{-1})
\end{eqnarray}
as claimed. The eigenvalues of $t^{(e)} (k;w_e)$ are given by
\begin{eqnarray}
\fl \qquad \mu_{1,2}(k) &= \frac{1\pm \frac{1}{2ik}\left|E'(L,k)+ikE(L,k)\right|}{e^{-ikL}- \frac{1}{2ik}\left( E'(L,k)-ikE(L,k)\right)} \nonumber  \\
&= e^{ikL} + \frac{1}{2ik}\frac{e^{ikL}\left(E'(L,k)-ikE(L,k)\right)\pm \left|E'(L,k)+ikE(L,k)\right|}{e^{-ikL}- \frac{1}{2ik}\left( E'(L,k)-ikE(L,k)\right)}
\end{eqnarray}
Let
\begin{eqnarray}
\tilde{E}_{1,2}(k):=\frac{1}{2ik}\frac{e^{ikL}\left(E'(L,k)-ikE(L,k)\right)\pm \left|E'(L,k)+ikE(L,k)\right|}{e^{-ikL}- \frac{1}{2ik}\left( E'(L,k)-ikE(L,k)\right)}
\end{eqnarray}
then
\begin{eqnarray}
 \left| \tilde{E}_{1,2}(k) \right| < \frac{\tilde{C}}{k}
&&\qquad\qquad\qquad  \left| \frac{\partial}{\partial k}\tilde{E}_{1,2}(k) \right| < \frac{\tilde{C}}{k}
\end{eqnarray}
from the bounds on $E(L,k)$ and $E'(L,k)$. Here $\tilde{C}$ can be computed from $C$, it also depends on the potential but not on $k$ or $x$.

Thus
\begin{eqnarray}
 \left.\frac{\partial }{\partial \varepsilon}\right|_{\varepsilon = 0}|\mu_{1,2}(k+i\varepsilon)|^2 =-2L + O(k^{-1})
\end{eqnarray}
which is smaller then zero for large enough values of $k$.
\end{proof}

\begin{remark}
This lemma still holds in the case of a $\delta$-potential as will be shown in section \ref{section:illustrate}. While one could use
this lemma to derive an explicit bound for $K$ the resulting bound would be very crude.

In the case of no potential or a positive $\delta$-potential one can choose $K=0$. For a negative $\delta$-potential
the transition matrix $t^{(e)}(k;w_e)$ is not subunitary above the real axis for small values of $k$ and an explicit value for the optimal
value of $K$ will be computed.
\end{remark}

\section {The secular equation and trace formula for graphs with edge potentials}
\label{section:trace}

We will now derive the secular function $\zeta_w$ for a quantum graph with potential. Each edge will be treated as a scatterer similar to
the vertices, using the transition matrix derived in the previous section.
The matrix $\Sigma$ controls the scattering at the vertices of the graph, and it is not affected by the introduction of the potential.
On all the vertices Kirchhoff-Neumann boundary conditions are imposed.


Using the $2E \times 2E$ matrices  $\Sigma$ and $T(k;w)$
  (\ref{def Sigma},\ref{def T}) we recall
\begin{eqnarray}
 S_w(k):=\Sigma \cdot T(k;w) \ .
\end{eqnarray}
Now the conditions on the coefficients $a_d$ and $b_d$ can be written as
\begin{eqnarray}
 S_w(k) \left( \begin{array}{c} b_1 \\ \vdots \\ b_{D} \end{array}\right)  = \left(\begin{array}{c} b_1 \\ \vdots \\ b_{D} \end{array}\right)
\end{eqnarray}
In other words, the eigenvalues of $\Delta_w$ on the graph correspond to the values of $k$ where $S_w(k)$ has eigenvalue $1$, including multiplicities.
The eigenvector defines the eigenfunction in terms of the $\psi^{\pm}$. This argument does not hold for $k=0$ as the functions $\psi^{\pm}$
are not linearly independent with our normalization.

Note that $S_w(k)$ is unitary on the real axis and subunitary above the real axis for large enough values of $k$ by
lemma \ref{subunitary} and the fact that $\Sigma$ is unitary and independent of $k$.

To prove theorem \ref {theorem-secular} recall  the definition (\ref {secular}) of the secular function
\begin{eqnarray}
 \zeta_w (k) := (\det S_w(k))^{-\frac{1}{2}}\det (Id_{2E}- S_w(k)) \ .
\end{eqnarray}
This function is zero whenever $S_w(k)$ has eigenvalue $1$, the prefactor is always nonzero because $S_w(k)$ is unitary. It is inserted
to make $\zeta_w(k)$ real on the real axis.


To prove theorem \ref {theorem-trace} one observes first that the function $\zeta_w$ can be extended holomorphically in $k$ into a
neighbourhood of the real line because
the end values $\psi^{\pm}(L)$ depend holomorphically on $k$, \cite{PoeschelTrubowitz87}. Moreover,
by the Schwarz reflection principle, the function $\zeta_w(k)$ satisfies
\begin{eqnarray}
 \zeta_w\left(\overline{k}\right)=\overline{\zeta_w(k)} .
\end{eqnarray}

The spectral density is derived using the Cauchy argument principle. Let $\{k_n^2\}_n$ denote the sequence of
positive eigenvalues. In order to expand the spectral density into a sum of periodic orbits, the matrix $S_w(k)$ needs to be subunitary above
the real axis. By lemma \ref{subunitary} this is only true for $k$ sufficiently large, therefore only eigenvalues
larger than some constant $K$ are counted. Negative eigenvalues, ie imaginary values of $k$ are also ignored.  Assume $K^2$ is not an eigenvalue. Let
\begin{eqnarray}
C_{\varepsilon, K}:= \left\{ z\in \C \mid -\varepsilon \le Im(z) \le \varepsilon, Re(z)> K \right\}
\end{eqnarray}
and let $\varphi$ be a test function, then
\begin{eqnarray}
\fl \qquad \sum_{k_n>K }^{\infty}\varphi(k_n)
&= \frac{1}{2\pi i}\lim_{\varepsilon \ra 0}\int_{\partial C_{\varepsilon,K}}\varphi(k)\frac{\partial}{\partial k}\ln (\zeta_w(k))dk \nonumber \\
&= -\frac{1}{\pi }\lim_{\varepsilon \ra 0}Im\int_{K}^{\infty}\varphi(k+i\varepsilon)\frac{\partial}{\partial k}\ln (\zeta_w(k+i\varepsilon))dk
\end{eqnarray}
where we used the Schwartz reflection principle. The integral over the interval $(-\varepsilon, \varepsilon)$ vanishes
in the limit $\varepsilon \ra 0$ because $K^2$ is not an eigenvalue.

Plugging in the definition of $\zeta_w(k)$ gives
\begin{eqnarray}
\fl \qquad \sum_{k_n>K }^{\infty}\varphi(k_n)
=& \frac{1}{2\pi }Im\int_{K}^{\infty}\varphi(k)\frac{\partial}{\partial k}\ln (\det S_w(k))dk\nonumber \\
& -\frac{1}{\pi }\lim_{\varepsilon \ra 0}Im\int_{K}^{\infty}\varphi(k+i\varepsilon)\frac{\partial}{\partial k}\ln \det (Id_{2E}- S_w(k+i\varepsilon))dk
\end{eqnarray}
In the first integral the limit $\varepsilon \ra 0$ commutes with the integral because $S_w(k)$ is unitary and thus the function inside
the integral does not have any poles on the real axis.

Let $\Theta_w(k) = \frac{1}{i} \ln \det T(k;w)$, then $\Theta_w(k)$ is real on the real axis
because $T(k;w)$ is unitary. Then
\begin{eqnarray}
 \frac{\partial}{\partial k}\ln (\det S_w(k)) =  i\frac{\partial \Theta_w(k)}{\partial k}
\end{eqnarray}
because $\Sigma$ is $k$-independent.

If $K$ is sufficiently large, by lemma \ref{subunitary}, the matrix $S_w$ is subunitary above the real axis, so it admits the expansion
\begin{eqnarray}
 \ln (\det (Id_{2E}- S_w(k + i\varepsilon))) = - \sum_{n=1}^{\infty}\frac{1}{n}\tr( S_w^n(k+i\varepsilon))
\end{eqnarray}
This gives
\begin{eqnarray}
\fl \qquad \sum_{k_n>K }^{\infty}\varphi(k_n)
=& \frac{1}{2\pi }\int_{K}^{\infty}\varphi(k) \frac{\partial \Theta_w(k)}{\partial k}dk\nonumber \\
& +\frac{1}{\pi }\lim_{\varepsilon \ra 0}Im\int_{K}^{\infty}\varphi(k+i\varepsilon)\frac{\partial}{\partial k}\sum_{n=1}^{\infty}\frac{1}{n} \tr( S_w^{n}(k+i\varepsilon))dk
\end{eqnarray}
The $\tr(S_w^n)$ terms can be interpreted as a sum over periodic orbits by using the identity.
\begin{eqnarray}
\sum_{n=1}^{\infty} \frac{\tr (\Sigma T(k;w))^n}{n}
= \sum_{n=1}^{\infty} \frac{1}{2n}\tr \left(\begin{array}{cc} 0 & \Sigma \\ T(k;w) & 0 \end{array}\right)^{2n}
\end{eqnarray}
Each edge in the original quantum graph is seen as a scatterer with scattering matrix given by the transition matrix $t^{(e)}(k;w_e)$
corresponding to the potential on that edge. Thus the $\tr(S_w(k)^n)$ terms are counting the periodic orbits of the bipartite graph $\mathcal{G}^*$
built from the original graph $\mathcal{G}$ by inserting an extra vertex on each edge. Let $PO^*$ denote the set of periodic orbits of
$\mathcal{G}^*$, $p \in PO^*$. Let
$n_p$ be the topological length of the periodic orbit $p$, that is the number of edges it traverses. Denote by $\tilde{p}$ the primitive
periodic orbit that $p$ is a repetition of. Then
\begin{eqnarray}
  Im \frac{\partial}{\partial k}
\sum_{n=1}^{\infty}\frac{1}{n}\tr \left( S_w(k)^n\right)
= \sum_{n=1}^{\infty}\frac{1}{n}\sum_{p \in PO^*, n_p=n} n_{\tilde{p}} Im \frac{\partial}{\partial k}\prod_{d \in p}\tau_{dd'}(k)
\end{eqnarray}
Here $d'$ is the bond that follows $d$ in $p$. The coefficients $\tau_{dd'}(k)$ are of two different types. If the vertex between them is a
vertex $v$ of the original graph then $\tau_{dd'}(k)=\sigma^{(v)}_{dd'}$ is the vertex scattering coefficient at that vertex from the
$\Sigma$ matrix, (which  is independent of $k$). If the vertex between them corresponds to one of the edges $e$ of the
original graph, then $\tau_{dd'}(k)=t^{(e)}_{dd'}(k)$ from the corresponding edge transition matrix.

The factor $ n_{\tilde{p}}$ comes from the fact that each periodic orbit is counted once for each starting vertex.

Theorem \ref {theorem-trace} follows once we identify
\begin{eqnarray}
\label{A-def}
 \mathcal{A}_p(k) =\frac{n_{\tilde{p}}}{n_p}Im \frac{\partial}{\partial k}\prod_{d \in p}\tau_{dd'}(k)
\end{eqnarray}

The trace formula can also be written in distribution form
\begin{eqnarray}
\sum_{k_n>K}^{\infty}\delta_{k_n}(k) =  \frac{1}{2\pi} \frac{\partial\Theta_w(k)}{\partial k}
+ \frac{1}{\pi}\sum_{p \in PO^*} \mathcal{A}_p(k)
\end{eqnarray}

\section {Example -- edges dressed with $\delta$-potentials}
\label{section:illustrate}

To illustrate the formal results, we discuss here  a graph dressed with  $\delta$-potential on the edges. The solutions $\psi^{\pm}$ are explicit so one can write down an explicit trace formula. This example shows clearly how an edge with a potential acts as a scatterer. Note that a $\delta$-potential can also be modelled by
introducing a vertex with suitable boundary conditions, there is an exact trace formula for these cases, see \cite{KottosSmilansky99} or
\cite{BolteEndres09}. The results here recover the ones in these two papers.

First consider a single edge parametrized as $[0,L]$, look at
\begin{eqnarray}
 -\psi''(x) + D \delta_{x_0}(x)\psi(x)&=k^2\psi(x)
\end{eqnarray}
This implies that $\psi'$ has a jump discontinuity at $x_0$ namely
\begin{eqnarray}
  \lim_{\varepsilon \ra 0}(\psi'(x_0+\varepsilon) - \psi'(x_0-\varepsilon)) = D\psi(x_0)
\end{eqnarray}
Here $D$ is some real parameter that measures the magnitude and $x_0 \in [0,L]$ is the location of the $\delta$-potential.

Assume for now that $k$ is real, so only  positive eigenvalues are considered. The two individual solutions are then of the form
\begin{eqnarray}
  \psi^+(x) &= \cases{
            e^{-ikx} & $x<x_0$\\
	    \frac{D}{2ik}e^{-2ikx_0}e^{ikx} + (1-\frac{D}{2ik})e^{-ikx} & $x>x_0$ \\} \nonumber \\
\psi^-(x) &= \cases{
            e^{ikx} &  $x<x_0$\\
	    (1+\frac{D}{2ik})e^{ikx} -\frac{D}{2ik}e^{2ikx_0}e^{-ikx} &$ x>x_0$\\}
\end{eqnarray}
where the coefficients in the linear combination where computed from the continuity of the solutions and the jump in their derivatives.

The Wronskian is then $W(\psi^+,\psi^-)=2ik$ on the entire edge.
In terms of the edge $\hat{d}$ with reverse orientation the $\delta$-potential is situated
at $L-x_0$, so the transition matrix is
\begin{equation}
\label{transition_delta_potential}
 t(k,D\delta_{x_0}) = \frac{1}{ (1-\frac{D}{2ik})e^{-ikL} }
 \left(\begin{array}{cc} 1 & \frac{D}{2ik}e^{-2ikx_0}e^{ikL}\\
 \frac{D}{2ik}e^{2ikx_0}e^{-ikL}   &  1
 \end{array}\right)
\end{equation}
with eigenvalues
\begin{eqnarray}
 \mu_1(k) = e^{ikL} \qquad\qquad\qquad \mu_{2}(k)=\frac{2ik+ D}{ 2ik-D }e^{ikL}
\end{eqnarray}

\begin{lemma}
\label{sufficiently_large_energy}
 If the $\delta$-potential is positive, $D \ge 0$, the matrix $t(k+i\varepsilon;D\delta_{x_0})$ is subunitary for all real $k$. If $D<0$ it is only subunitary for sufficiently large values of $k$, namely for
\begin{eqnarray}
 k^2 > -\frac{D}{L}- \frac{D^2}{4}
\end{eqnarray}

\end{lemma}
\begin{proof}
Let $\varepsilon > 0$ then
\begin{eqnarray}
 \mu_1(k+i\varepsilon) = e^{ikL}e^{-\varepsilon L}
\end{eqnarray}
so
\begin{eqnarray}
| \mu_1(k+i\varepsilon) | < | \mu_1(k)|=1
\end{eqnarray}
For the second eigenvalue
\begin{eqnarray}
 \mu_{2}(k+i\varepsilon)=\frac{2ik-2\varepsilon+ D}{ 2ik-2\varepsilon-D }e^{ikL}e^{-\varepsilon L}
\end{eqnarray}
thus
\begin{eqnarray}
 \left.\frac{d}{d \varepsilon}\left|\mu_2(k+i\varepsilon) \right|^2\right|_{\varepsilon = 0} =- \frac{8D}{D^2+4k^2}-2L
\end{eqnarray}
which means
\begin{eqnarray}
 \left|\mu_2(k+i\varepsilon) \right| < 1 &\qquad \Leftrightarrow &\qquad L > \frac{-D}{(D/2)^2+k^2}
\end{eqnarray}
This is trivially satisfied if $D \ge 0$ but not true in general if $D< 0$.
\end{proof}

\begin{remark}
 The matrix $t(k; D\delta_{x_0})$ is only subunitary above the real axis for large enough values of $k$. This is the same
situation as in lemma \ref{subunitary} for potentials in $L^2([0,L])$.
\end{remark}

The trace formula contains terms of the form $\frac{\partial}{\partial k}\ln \det t(k;D\delta_{x_0})$ which  can be explicitly computed using
\begin{eqnarray}
 \mu_1(k) = e^{ikL} &&\qquad\qquad \mu_{2}(k)=\frac{2ik+ D}{ 2ik-D }e^{ikL} =e^{-2i\arctan(\frac{D}{2k})+ikL}
\end{eqnarray}
so that
\begin{eqnarray}
 \frac{1}{i} \frac{\partial}{\partial k}\ln \det t(\delta;k)= 2L +\frac{2D}{4k^2+D^2}
\end{eqnarray}

\begin{corollary}
 The trace formula with a $\delta$-potential of strength $D_e$ on each edge reads
 \begin{eqnarray}
\sum_{k_n>K}^{\infty}\delta_{k_n}(k) = \frac{\mathcal{L}}{\pi}+ \frac{1}{\pi}\sum_{e=1}^{E}\frac{D_e}{4k^2+D_e^2}
+ \frac{1}{\pi}\sum_{p \in PO^*} \mathcal{A}_p(k)
\end{eqnarray}
with coefficients
\begin{eqnarray}
 \mathcal{A}_p(k) =\frac{n_{\tilde{p}}}{n_p}Im \frac{\partial}{\partial k}\prod_{d \in p}\tau_{dd'}(k)
\end{eqnarray}
and the $\tau_{dd'}(k)$ from the matrix $\Sigma T(k;D\delta_{x_0})$ with $t^{(e)}(k;D_e\delta_{x_0})$ computed above in equation
\ref{transition_delta_potential}. Here $K$ has
to be chosen so that the inequality in lemma \ref{sufficiently_large_energy} is satisfied on all edges.
\end{corollary}

\section{The high-energy limit}
\label{section:high_energy}

Consider the high-energy limit, that is, assume that $k$ is real and $k^2 >> ||w||_{L^{\infty}(G)}$. The exact solutions will be approximated by the WKB-method. This requires some estimates on the WKB approximation so the main steps in its derivation will be recalled. The goal is to find an approximate solution to
\begin{eqnarray}
 -\psi''(x) + w(x)\psi(x) = k^2\psi(x)
\end{eqnarray}
for large values of $k$. Set
\begin{eqnarray}
\label{s-definition}
 p(x):=\sqrt{k^2-w(x)} &&\qquad\qquad\qquad s(x) :=\int_0^x p(x')dx'
\end{eqnarray}
and then use the substitution
\begin{eqnarray}
 \psi(x)=: p^{-1/2}(x)\eta(s(x))
\end{eqnarray}
This is well defined as $p(x)>0$ by the assumptions. The equation then transforms to
\begin{eqnarray}
\label{wkbform}
\fl \qquad \frac{\partial^2 }{\partial s^2}\eta(s) + \eta(s) = - \left( \frac{1}{4} w''(x(s))p^{-4}(x(s))  +\frac{5}{16} w'(x(s))^2p^{-6}(x(s)) \right) \eta(s)
\end{eqnarray}
Let $\chi(x):= \frac{1}{4} w''(x)p^{-4}(x)  +\frac{5}{16} w'(x)^2p^{-6}(x)$ to simplify notation.
Equation (\ref {wkbform}) is solved iteratively by a power series $\eta(s):=\sum_{j=0} \eta_j(s)$. The two initial solutions are $\eta_0(s) := e^{\pm is}$
and the general term is the solution of
\begin{eqnarray}
 \frac{\partial^2}{\partial s^2}\eta_j(s) + \eta_j(s) = - \chi(x(s)) \eta_{j-1}(s)
\end{eqnarray}
This is a second order linear differential equation which can be solved explicitly through variation of constants.
\begin{eqnarray}
 \eta_j(s) = -e^{is}\int_0^se^{-2is'}\int_0^{s'}e^{is''}\chi(x(s''))\eta_{j-1}(s'')ds''ds'
\end{eqnarray}

This process converges if $\chi$ is sufficiently small, which is satisfied for $k$ sufficiently large.
The WKB-approximation is simply $\eta_0$, written in the $x$ variable this yields the familiar
\begin{eqnarray}
 \psi^{\pm}_{WKB}(x) := \sqrt{\frac{p(0)}{p(x)}}e^{\mp i\int_0^{x}p(x')dx'}
\end{eqnarray}
where the constant $\sqrt{p(0)}$ term was inserted for normalization.

\begin{lemma}
 On each edge the WKB-solutions differ from the exact solutions by an error term that decays like $k^{-2}$ for $k$ sufficiently large.
\begin{eqnarray}
 \left| \psi^{\pm}(x) -\psi_{WKB}^{\pm}(x) \right| &< \frac{C}{k^2} \nonumber \\
 \left| (\psi^{\pm})'(x) -(\psi_{WKB}^{\pm})'(x) \right| &< \frac{C}{k}
\end{eqnarray}
with $C$ independent of $k$ and $x \in [0,L]$.
\end{lemma}
\begin{proof}
 For $k$ large, $p(x) \sim k$ so $s(x) \sim k$ and $\chi(x) \sim k^{-4}$. The iteration then gives $\eta_j(x) \sim k^{-2j}$ and
$(\eta_j)'(x) \sim k^{-2j+1}$ so the  results follow from
\begin{eqnarray}
 \psi^{\pm}(x) -\psi_{WKB}^{\pm}(x)=  \sum_{j=1}\eta_j(x)
\end{eqnarray}
\end{proof}

On top of approximating the exact solution by the WKB-solution, the derivatives of the WKB-solution will be approximated as follows.
\begin{eqnarray}
 (\psi_{WKB}^{\pm})'(x) = \mp i k\psi_{WKB}^{\pm}(x)+O(k^{-1})
\end{eqnarray}
Thus
\begin{eqnarray}
 \psi_{WKB}^{\pm}(0)=1 &&\qquad\qquad\qquad (\psi_{WKB}^{\pm})'(0) = \mp ik + O(k^{-1})
\end{eqnarray}
Both layers of approximation have an error term that is two orders smaller than the main term.

The transition matrix for the WKB-approximation can now be computed explicitly.
\begin{eqnarray}
\label{WKB_transition}
 &t(k;w)\nonumber \\
 = & t^{WKB}(k;w) +O(k^{-2}) \nonumber \\
 = &  \left(\begin{array}{cc} (\psi^+_{WKB}(L))^{-1} & 0 \\ 0 & (\psi^+_{WKB}(L))^{-1} \end{array}\right)  +O(k^{-2}) \nonumber \\
 = &  \left(\begin{array}{cc} e^{i\int_0^Lp(x)dx} & 0 \\ 0 & e^{i\int_0^Lp(x)dx} \end{array}\right) +O(k^{-2})
\end{eqnarray}

This information will now be used to study the behaviour of the trace formula in the high energy limit. Notice that the function $s$ defined
in (\ref{s-definition}), now seen as a function of $k$, corresponds to the classical action along an edge of the graph:
\begin{equation}
s_{d}(k)=s_{\hat d}(k) := \int_{0}^{L_d}\sqrt{k^2-w_d(x)}dx
\end{equation}
Recall that $\Theta_w(k) = \frac{1}{i} \ln \det T(k;w)$,  which can be approximated to $O(k^{-2})$ by
\begin{eqnarray}
\Theta_w(k)
=  \Theta_w^{WKB}(k)+O(k^{-2})
=  \sum_{d\in \mathcal{D}}s_{d}(k) +O(k^{-2})
\end{eqnarray}

Note that the trace formula as stated in equation (\ref{exact_trace}) contains the derivative of the $\Theta_w(k)$, which in the WKB approximation reads
\begin{eqnarray}
\frac {\partial\ }{\partial k} \Theta_w^{WKB}(k) = \sum_{d\in \mathcal{D}}\int_{0}^{L_d}\frac{k dx}{\sqrt{k^2-w_d(x)}} + \frac{\partial}{\partial k} O(k^{-2}).
\end{eqnarray}
Each of the integrals above can be identified as the time it takes a particle to traverse the edge (in the system of units used here, ``time" is expressed in units of length). In the physics literature, the derivative of $\Theta_w$ is known as the Wigner delay time which will be denoted here by $\mathcal{T}(k)$. The result above can be interpreted as the Wigner delay time for a graph. In the high energy limit, it approaches the total length of the graph.

An important consequence of the WKB approximation to the edge transition matrix (\ref{WKB_transition}) is that reflections by the edge potential can be neglected to order $\frac{1}{k^2}$ . This fact allows us to neglect the contribution to the trace formula from all the periodic orbits in the graph that involve a reflection on an edge.   In other words, only the periodic orbits from the original graph $\mathcal{G}$ survive. Their weights  $\mathcal{A}_p(k)$ split into a vertex scattering part and the contribution from the edges.
\begin{eqnarray}
\mathcal{A}_p(k)
 =\mathcal{A}_p^{WKB}(k) + \frac{\partial}{\partial k} O(k^{-2}) \ .
 \end{eqnarray}
 and
 \begin{eqnarray}
 \mathcal{A}_p^{WKB}(k)  =  \frac{\partial}{\partial k} \left( \frac{n_{\tilde{p}}}{n_p}Im \prod_{v \in p}\sigma_{dd'}^v\prod_{d \in p} e^{i\int_0^{L_d}\sqrt{k^2-w_d(x)}dx}\ .
 \right)
\end{eqnarray}
Here the $\sigma_{dd'}^v= -\delta_{bb'}+\frac{2}{\deg(v)}$ are from the Kirchhoff-Neumann boundary conditions as before. To cast the trace formula in a
form similar to the semi-classical trace formula  \cite{Gutzwiller90}, it is convenient to define the following quantities for each periodic orbit of topological length $n_p$ which consists of  $r_p$ repetitions of a primitive periodic orbit of length $n_{\tilde p}$:

\noindent Let $A_{\tilde p} = |\prod_{v \in {\tilde p}}\sigma_{dd'}^v| $ and let $\nu_{\tilde p}$ be the number of back scatters (where $\sigma_{dd'}^v  <0$). Let $ S_{\tilde p}(k) = \sum_{d\in {\tilde p}} s_d(k)$ which is the classical action along the periodic orbit, and let $ T_{\tilde p}(k)= \frac{\partial }{\partial k}  S_{\tilde p}(k)$ denote the classical period. Then,
\begin{eqnarray}
 \mathcal{A}_p^{WKB}(k)  =  T_{\tilde p}(k)A_{\tilde p}^{r_p} \cos((S_{\tilde p}(k)+\pi \nu_{\tilde p})r_p)\ .
 \end{eqnarray}

 The analogy with the Gutzwiller  trace formula is completed when  $A_{\tilde p}^2$ is identified as the classical probability to survive along the period orbit, and $\nu_{\tilde p}$ the index replacing the Maslov index.  Again to zeroth order this
simplifies to the no potential situation.

Finally,
\begin{eqnarray}
 \sum_{k_n >K}^{\infty}\varphi(k_n)
&=  \frac{1}{2\pi }\int_{K}^{\infty}\varphi(k) \mathcal{T}(k) dk \nonumber \\
 &+\frac{1}{\pi }\int_{K}^{\infty}  \varphi(k) \sum_{p \in PO} T_{\tilde p}(k) A_{\tilde p}^{r_p} \cos((S_{\tilde p}(k)+\pi \nu_{\tilde p})r_p)dk \nonumber \\
 &+   \int_K^{\infty}\varphi(k)\frac{\partial }{\partial k}E(k)dk
 \end{eqnarray}
 where $E(k)$ is the error term, it decays as $E(k)=O(k^{-2})$.

\begin{remark}
Note that the term involving the Euler
characteristic from the classic trace formula is not visible because the formula here is only valid for large values of $k$ and this term
would appear as a $\delta$ distribution located at $k=0$.
\end{remark}

\section*{Acknowledgements}

\noindent This work was supported by the Einstein (Minerva) Center at the Weizmann Institute and
the Wales Institute of Mathematical and Computational Sciences)
(WIMCS). Grants from EPSRC (grant EP/G021287), BSF (grant 2006065) and ISF (grant 166/09) are acknowledged.
We would also like to thank Dr Rami Band for a critical reading of the manuscript and many helpful comments.

\section*{References}
\bibliographystyle{amsalpha}
\addcontentsline{toc}{section}{Literature}
\bibliography{literatur}

\end{document}